\documentclass[a4paper,11pt]{article}
\usepackage{amsfonts}
\usepackage{amsmath,cite,tikz,multirow,amssymb,mathrsfs,latexsym,bm,inputenc,longtable,amsfonts}
\usepackage{graphicx}
\usepackage{tabularx}
\usepackage{array}
\usepackage{float}
\usepackage[english]{babel}
\usepackage{enumitem}
\usepackage{colortbl}
\usepackage{xcolor}
\usepackage{hhline}
\usepackage{booktabs}
\begin{document}

\newcommand{\qed}{\hphantom{.}\hfill $\Box$\medbreak}
\newcommand{\Proof}{\noindent{\bf Proof \ }}

\newtheorem{theorem}{Theorem}[section]
\newtheorem{lemma}[theorem]{Lemma}
\newtheorem{corollary}[theorem]{Corollary}
\newtheorem{remark}[theorem]{Remark}
\newtheorem{example}[theorem]{Example}
\newtheorem{definition}[theorem]{Definition}
\newtheorem{construction}[theorem]{Construction}
\newtheorem{proposition}[theorem]{Proposition}

\title{\large{\bf Local unitary equivalence of orthogonal arrays and \\ related linear codes\footnote{Supported by the Science Research Project of the Hebei Education
Department (Grant No. JCZX2026032), the Natural Science Foundation of Hebei Province (Grant No. A2025205023), the National Natural Science Foundation of China (Grant No. 12601632, 12571493).}}}

\author{Miaomiao Zheng$^{a}$, Yajuan Zang$^{b}$, Xiuling Shan$^{a}$, Zihong Tian$^{a}$\thanks{Corresponding author. E-mail address: tianzh68@163.com}  \\
\small$^a$ School of Mathematical Sciences, Hebei Normal University,
          Shijiazhuang 050024, P. R. China          \\
\small $^b$ School of Teacher Education, Hebei Normal University,
          Shijiazhuang 050024,  P. R. China}

\date{}

\maketitle

\noindent {\bf Abstract:}
Orthogonal arrays (OAs) are combinatorial configurations with applications in experimental design, error-correcting codes, and quantum information.
Local unitary (LU) equivalence provides a natural framework for classifying multipartite entangled states.
Using the correspondence between OAs and quantum states, Goyeneche and \.{Z}yczkowski [Phys. Rev. A, 2014, 90: 022316] posed the problem of determining when OAs are $LU$ equivalent.
In this paper, we construct OAs from generator matrices and establish conditions for Fourier-based $LU$  equivalence of OAs and irredundant orthogonal arrays (IrOAs). For prime alphabets, we identify the Fourier partners of the linear OAs considered here with the arrays of their corresponding dual codes. This gives an explicit connection between $LU$  equivalence and coding theory. We also identify families of $LU$ equivalent OAs arising from specific classes of linear codes.
\vspace{0.2cm}

\noindent {\bf Keywords}: local unitary equivalence,  orthogonal arrays,  irredundant orthogonal arrays, linear codes


\section{Introduction}

An $r\times N$ array $A$ with entries from a $d$-element set $S$ is an \emph{orthogonal array} (OA) with $r$ runs, $N$ factors, $d$ levels, and strength $k$, denoted by OA$(r,N,d,k)$, if every $r\times k$ subarray contains each $k$-tuple of symbols from $S$ exactly $\lambda$ times as a row, where $\lambda=r/d^k$. An OA$(r,N,d,k)$ is \emph{irredundant}, denoted by IrOA$(r,N,d,k)$, if all rows of every $r\times(N-k)$ subarray are distinct.

Orthogonal arrays were introduced by Rao~\cite{Rao} in 1947 and have important applications in diverse fields. For example, they can be used to define generalized forms of Latin hypercube sampling and lattice sampling in fractional factorial experiments~\cite{Owen}, and to generate substitution sequences and permutation mappings for image encryption~\cite{xuming}. In coding theory, the rows of an orthogonal array can be viewed as codewords~\cite{Hedayat1999}.

Orthogonal arrays have also been widely used to study quantum entanglement~\cite{Goyeneche2014,Li,Goyeneche2015,N-Arul}. Quantum entanglement is a fundamental feature that distinguishes quantum systems from classical systems. A pure multipartite state is called $k$-uniform if every reduction to $k$ parties is maximally mixed. Such states have applications in quantum teleportation, quantum key distribution, dense coding, and error-correcting codes~\cite{Raissi,Du}. In particular, IrOAs provide a useful method for constructing $k$-uniform states~\cite{Goyeneche2014,Zang1,Zang2,Pang1,Chen}. Specifically, consider an
  OA$(r,N,d,k)$ with distinct rows, written as
 \begin{center}
$A=\left(
\begin{array}{cccc}
s^1_1 & s^1_2 & \cdots & s^1_N \\
 s^2_1 & s^2_2 & \cdots & s^2_N \\
 \vdots & \vdots & \cdots & \vdots\\
 s^r_1 & s^r_2 & \cdots & s^r_N \\
 \end{array}
\right),$
\end{center}
where $ s^j_i\in \{0,1,\dots,d-1\}$ for $j=1,2,\dots,r$ and $i=1,2,\dots,N$. Then
$$|\psi\rangle= \frac{1}{\sqrt{r}}(|s^1_1  s^1_2  \cdots s^1_N\rangle +
 |s^2_1  s^2_2  \cdots s^2_N\rangle +\cdots +
 |s^r_1   s^r_2   \cdots s^r_N\rangle)$$
is a normalized pure quantum state in $ (\mathbb{C}^d)^{\otimes N},$ where $ \mathbb{C}^d $ is the complex vector space of dimension $d$,  $|s^j_1  s^j_2  \cdots s^j_N\rangle$ denotes the product state $|s^j_1\rangle \otimes |s^j_2\rangle \otimes \cdots \otimes |s^j_N\rangle$ and $\otimes$ denotes
the tensor product. Specifically, $|i\rangle$ is a unit column vector in $\mathbb{C}^d$ with its $(i+1)$-th component equal to $1$, and $\{|0\rangle, |1\rangle, \ldots, |d-1\rangle\}$ is an orthonormal basis for $\mathbb{C}^d$.
Furthermore, if $A$ is an IrOA$(r, N, d, k)$, then $|\psi\rangle$ is a $k$-uniform state \cite{Goyeneche2014}.

A central problem in quantum information theory is to classify the different forms of multipartite entanglement. The
nonlocal nature of entanglement motivates classifications under local operations. Two pure states are equivalent under stochastic local operations and classical communication (SLOCC) if they can be converted into each other with nonzero probability of success~\cite{Dur}. Three-qubit pure states have six SLOCC classes, whereas four-qubit pure states have infinitely many SLOCC classes that can be organized into nine families~\cite{Dur,Verstraete}. The classification becomes more complicated as the number of parties increases. Local unitary (LU) equivalence is a special case of SLOCC equivalence that preserves the degree of entanglement of quantum states \cite{Liu,Zhang}. It therefore yields a finer classification.
Two $N$-partite pure states $|\phi\rangle$ and $|\psi\rangle$ are \emph{LU equivalent} if there exist unitary operators $U_1,U_2,\dots,U_N$ such that $|\phi\rangle=(U_1\otimes U_2\otimes\cdots\otimes U_N)|\psi\rangle$~\cite{Kraus,Liu}.

Using the correspondence between quantum states and orthogonal arrays, Goyeneche and \.{Z}yczkowski~\cite{Goyeneche2014} posed the problem of determining when OAs are $LU$  equivalent. Two orthogonal arrays with distinct rows, $\mathrm{OA}(r,N,d,k)$ and $\mathrm{OA}(r',N,d,k')$ are \emph{LU equivalent},   if their corresponding quantum states are $LU$  equivalent.
The arrays must have the same number of columns and levels, but their numbers of rows and strengths may differ. Throughout the paper, $LU$  equivalence refers to the normalized, positive equal-amplitude states associated with arrays having distinct rows; the order of the parties is fixed. An OA of strength $k$ also has every smaller positive strength, so a stated strength need not be maximal.
Du et al.~\cite{Du} proved that there exists an  $\mathrm{OA}(d^n,n+1,d,n)$, that is $LU$  equivalent to an $\mathrm{OA}(d,n+1,d,1)$ for any positive integer $n$.
Zhang et al.\cite{Zhang-Pang}  presented several methods for constructing infinitely many classes of $LU$  equivalent  OAs and proved that a linear
OA and its dual are $LU$  equivalent. Ramadas and Lakshminarayan~\cite{N-Arul} showed that varying the phases of the states associated with an IrOA$(r,N,d,k)$ yields infinitely many $LU$  equivalence classes when $r>Nd-(N-1)$.

In this paper, we establish conditions for $LU$  equivalence of OAs, focusing on transformations based on Fourier matrices, which we denote by $LU_F$. We construct OAs from generator matrices and investigate $LU_F$ equivalence of both OAs and IrOAs. Specially, we show that certain pairs of $LU_F$ equivalent OAs correspond to linear codes and their duals.

The remainder of this paper is organized as follows. Section 2 presents a construction of OAs from generator matrices and establishes conditions for LU$_F$ equivalence to OAs of strength $1$. Section 3 considers IrOAs of strength $1$ or $2$. Section 4 connects LU$_F$ equivalence of linear OAs with duality of linear codes and gives examples based on
MDS, Hamming, Reed--Muller, cyclic, and negacyclic codes. Section 5 concludes the paper.

We first introduce the notation used throughout the paper.
Let $\mathbb{N}^+$ be the set of positive integers.
For an integer $d\geq2$, let $[d]=\{0,1,\dots ,d-1\}$. $\mathbb{Z}_d$ represents the residue ring of integers modulo $d$.
$\mathbb{Z}_d^*$ is the set of invertible elements in $\mathbb{Z}_d$ with respect to multiplication. $\mathbb{Z}_d^n$ is the set of $n$-tuples over $\mathbb{Z}_d$. $m \mathbb{Z}_d$ denotes the multiset in which each element of $\mathbb{Z}_d$ occurs $m$ times.  $\textbf{0}_d$ and $\textbf{1}_d$ represent the length-$d$ vectors (with orientation determined by context)  consisting entirely of $0$s and entirely of $1$s, respectively.
$I_d$ is the identity matrix of order $d$.
$F_d$ is the Fourier matrix of order $d$.
$|A|$ denotes the determinant when $A$ is a square matrix and cardinality when $A$ is a finite set.    $B^T$ is the transpose of a matrix $B$. $\operatorname{rank}(C)$ represents the rank of the matrix $C$.
 $\delta_{ij}$ is the Kronecker delta, i.e., $\delta_{ij}=1$ if $i=j$ and $0$ otherwise. For two quantum states $|\psi\rangle$ and $|\phi\rangle$,
 $\langle\psi|\phi\rangle$ denotes the inner product and
$|\psi\rangle\langle\phi|$ denotes the outer product of $|\psi\rangle$ and $|\phi\rangle$.

\section{LU equivalence of OAs}
In this section, we construct OAs from generator matrices and establish conditions under which they are $LU$  equivalent to OAs of strength $1$.
Specially, we describe the relationship between the generator matrices of OAs that are connected by Fourier-based local unitary transformations.

 \subsection{A construction of OAs from generator matrices}

We begin with the required definitions.
Let $D$ be an $m\times n$ matrix over $\mathbb{Z}_d$, where $d\geq2$. Choose integer representatives for its entries. For $1\leq k\leq\min\{m,n\}$, a \emph{minor of order $k$} is the determinant of a $k\times k$ submatrix of this integer lift. Let $D_k$ be the nonnegative greatest common divisor of these minors, with $D_k=0$ if all of them vanish. Although $D_k$ can depend on the lift, the condition $\gcd(D_k,d)=1$ does not: it says that the minors generate the unit ideal modulo $d$. We write $D_k\in\mathbb{Z}_d^*$ for this condition. Over a composite modulus, it does not require any single minor to be a unit.

\begin{lemma}{\rm (\!\!\cite{Liyu})}\label{l1}
Let $D$ be an $m \times k$ matrix over $\mathbb{Z}_d$, where $d\geq2$ and $1\leq k\leq m$, and let $b$ range over $\mathbb{Z}_d^k$. Then the system $(x_1,x_2,\dots,x_m)D\equiv b\pmod{d}$ has $d^{m-k}$ solutions for every $b$ if and only if $D_k \in \mathbb{Z}_d^*$.
\end{lemma}

\begin{construction}\label{c}
 Let $G$ be an $m \times n$ matrix over $\mathbb{Z}_d$, where $d\geq2$ and $m,n\in\mathbb{N}^+$. Then the array obtained by evaluating $xG$ for all $x\in\mathbb{Z}_d^m$ is an $\mathrm{OA}(d^m, n, d, k)$ over $\mathbb{Z}_d$ if and only if every $m \times k$ submatrix $D$ of $G$ satisfies $D_k \in \mathbb{Z}_d^*$, where $1\leq k \leq \min\{m, n\}$.
\end{construction}
\Proof Let $C$ be the $d^m\times m$ matrix whose rows are all vectors in $\mathbb{Z}_d^m$, and set $A=CG$, with multiplication performed modulo $d$. Selecting $k$ columns of $A$ gives a subarray $CD$, where $D$ consists of the corresponding columns of $G$. A vector $y\in\mathbb{Z}_d^k$ occurs in this subarray once for each solution of $xD=y$. By Lemma~\ref{l1}, every $y$ occurs exactly $d^{m-k}$ times if and only if $D_k\in\mathbb{Z}_d^*$. Requiring this for every selection of $k$ columns proves both directions.\qed

This construction permits repeated rows. The associated equal-amplitude state is used only when the rows are distinct; in particular, this holds whenever $G=(I_m\mid G')$.

\begin{example}
There exists an $\mathrm{OA}(6^3, 4, 6, 2)$ over $\mathbb{Z}_6$.
\end{example}

\Proof Let
$$
G = \begin{pmatrix}
1 & 1 & 1 & 1 \\
1 & 2 & 3 & 4 \\
1 & 1 & 2 & 2
\end{pmatrix} $$
be a $3 \times 4$ matrix over $\mathbb{Z}_6$.
One can verify that for any $3 \times 2$ submatrix $D$, we have $D_2 \equiv 1 \pmod{6}$.
Let $C$ be the $6^3 \times 3$ matrix whose rows consist of all $3$-tuples over $\mathbb{Z}_6$.
Define $A = CG$.
It can be verified that for any $2$ columns of $A$, each ordered pair over $\mathbb{Z}_6$ appears exactly $6$ times. Therefore, $A$ is an $\mathrm{OA}(6^3, 4, 6, 2)$.  \qed

The following example shows that when $d \geq 2$, for any positive integers $m, n$ with $m \geq n$, we can construct an $\mathrm{OA}(d^m, n, d, k)$ from an $m \times n$ matrix, where $k\leq n \leq m$.

\begin{theorem}
For all integers $d\geq2$ and $1\leq k\leq n\leq m$, there exists an $\mathrm{OA}(d^m,n,d,k)$.
\end{theorem}

\Proof Define an $m \times n$ matrix over $\mathbb{Z}_d$ as follows, where $*$ denotes an arbitrary element of $\mathbb{Z}_d$.
$$
G = \begin{pmatrix}
1 & 0 & \cdots & 0 \\
* & 1 & \cdots & 0 \\
\vdots & \vdots & \ddots & \vdots \\
*  & *  & \cdots & 1 \\
* & * & \cdots & * \\
\vdots & \vdots & \ddots & \vdots \\
* & * & \cdots & *
\end{pmatrix}.$$
  For any $ m \times k$ submatrix $D$ of $G$, we can find a $k \times k$ invertible submatrix $D'$ of $D$,
$$
D' = \begin{pmatrix}
1 & 0 & \cdots & 0 \\
* & 1 & \cdots & 0 \\
\vdots & \vdots & \ddots & \vdots \\
* & * & \cdots & 1
\end{pmatrix}.$$
Thus, $D_k = 1$, and there exists an $\mathrm{OA}(d^m, n, d, k)$ over $\mathbb{Z}_d$ constructed from $G$ by Construction~\ref{c}.\qed

\subsection{LU equivalence of OAs}
In this subsection, we establish sufficient conditions under which an OA constructed as in Section 2.1 is $LU$  equivalent to an OA of strength $1$.
We begin with the required definitions and notation.

An \emph{elementary row operation} over $\mathbb{Z}_d$ swaps two rows, multiplies a row by a unit, or adds a multiple of one row to another. Invertible row operations on a generator matrix permute the input vectors $x\in\mathbb{Z}_d^m$ and therefore preserve its generated array up to row order.

Two OAs $A$ and $A'$ are \emph{combinatorially equivalent},  if they differ by row permutations, column permutations, and independent symbol permutations within columns. Column permutations and multiplication of columns by units preserve this equivalence. In contrast, adding one column to another need not preserve the OA strength or combinatorial equivalence and is not allowed in this argument.

Row permutations do not change the associated state, and symbol permutations act by local permutation unitaries. A column permutation relabels the parties and need not itself be a local unitary. We therefore distinguish combinatorial equivalence from $LU$  equivalence. A common column permutation applied to both arrays preserves any established $LU$  relation and commutes with $U_d^{\otimes n}$.

We explicitly assume a \emph{systematic} generator matrix
\[
G=(I_m\mid G'),\qquad G'=(g'_{ij})_{m\times(n-m)},
\]
in the results below. Such a form is available, using column permutations and invertible row operations, if $G$ has an $m\times m$ minor that is a unit. Full row rank suffices over a field, but an unspecified rank condition over $\mathbb{Z}_d$ does not justify this conclusion for composite~$d$.

  Let $F_d=(F_d(k,l))$ be the Fourier matrix of order $d$ as follows

$$
F_d=\left(\begin{array}{ccccc}
1& 1&1 & \cdots &1\\
1  & \omega &\omega^2 &\cdots &\omega^{d-1} \\
1  & \omega^2 &\omega^4 &\cdots &\omega^{2(d-1)} \\
1  & \vdots &\vdots &\ddots &\vdots\\
1  & \omega^{d-1} &\omega^{2(d-1)}&\cdots &\omega^{(d-1)^2} \\
\end{array}\right), $$
\noindent where $F_d(k,l)= \omega^{kl}$, $k,l \in [d]$, $\omega=\mathrm{e}^{\frac{2\pi \mathrm{i}}{d}}$ is a primitive $d$-th root of unity.
Define
$$U_d = \frac{1}{\sqrt{d}}F_d= \frac{1}{\sqrt{d}} \sum_{l,b=0}^{d-1} \omega^{lb} |l\rangle\langle b|.$$

The matrix $U_d$ is unitary. We call two states \emph{$LU_F$ equivalent} if one is obtained from the other by a tensor product of integer powers of $U_d$. This is an equivalence relation contained in $LU$ equivalence. All Fourier partners constructed below are obtained by the specific transformation $U_d^{\otimes n}$. We do not claim that this single transformation produces every $LU_F$ equivalent array.

\begin{lemma}\label{l2}
Let $d,t\in\mathbb{N}^+$ with $d\geq 2$, and let $a_1, a_2, \dots, a_t\in \mathbb{Z}_d.$ If at least one coefficient $a_i\in \mathbb{Z}_d^*$, $i=1,2,\dots,t$, then
$$
\{a_1 x_1 + a_2 x_2 + \cdots + a_t x_t : x_1, x_2, \dots, x_t \in \mathbb{Z}_d\} = d^{t-1} \mathbb{Z}_d.
$$

\end{lemma}
\Proof Without loss of generality, suppose $a_1 \in \mathbb{Z}_d^*$,
then
\begin{eqnarray*}
\{a_1 x_1 : x_1 \in \mathbb{Z}_d\} =   \mathbb{Z}_d.
\end{eqnarray*}
Moreover, $x + \mathbb{Z}_d = \{x+a: a\in \mathbb{Z}_d\}= \mathbb{Z}_d$ for any $x \in \mathbb{Z}_d$. Therefore, the multiset\vspace{0.3cm}

\hspace{3cm} $\{a_1 x_1 + a_2 x_2 + \cdots + a_t x_t : x_1, x_2, \dots, x_t \in \mathbb{Z}_d\} = d^{t-1}  \mathbb{Z}_d.$
 \qed

\begin{theorem}\label{t1}
Let $d, m, n, k \in \mathbb{N}^+$ with $d\geq2$, $k \leq m < n$, and let $G = (I_m \mid G')$ be an $m \times n$ matrix over $\mathbb{Z}_d$. If $G$ satisfies

 $(1)$ $D_k \in \mathbb{Z}_d^*$ for every $m \times k$ submatrix $D$ of $G;$

 $(2)$ each row of $G'$ contains at least one element in $\mathbb{Z}_d^*$,

\noindent then $G$ generates an $\mathrm{OA}(d^m, n, d, k)$ over $\mathbb{Z}_d$. Furthermore, the $\mathrm{OA}(d^m, n, d, k)$ generated by $G$ is $LU_F$ equivalent to an $\mathrm{OA}(d^t, n, d, 1)$, where $t = n - m$.
\end{theorem}
\Proof Put $t=n-m$ and let $A$ have rows $xG=(x,xG')$, where $x\in\mathbb{Z}_d^m$. Condition~(1) and Construction~\ref{c} give strength $k$, and the identity block ensures that the $d^m$ rows are distinct. Hence
\[
|\psi_A\rangle=d^{-m/2}\sum_{x\in\mathbb{Z}_d^m}|x,xG'\rangle.
\]
Write a vector in $\mathbb{Z}_d^n$ as $(a,b)$, with $a\in\mathbb{Z}_d^m$ and $b\in\mathbb{Z}_d^t$. Character orthogonality gives
\[
\sum_{v\in\mathbb{Z}_d}\omega^{uv}=
\begin{cases}d,&u=0\text{ in }\mathbb{Z}_d,\\0,&u\neq0\text{ in }\mathbb{Z}_d.\end{cases}
\]
Consequently, with all vector arithmetic performed modulo $d$,
\begin{align*}
U_d^{\otimes n}|\psi_A\rangle
&=d^{-(m+n)/2}\sum_{a\in\mathbb{Z}_d^m}\sum_{b\in\mathbb{Z}_d^t}
  \left(\sum_{x\in\mathbb{Z}_d^m}\omega^{x(a^T+G'b^T)}\right)|a,b\rangle\\
&=d^{-t/2}\sum_{b\in\mathbb{Z}_d^t}|-b(G')^T,b\rangle.
\end{align*}
The last expression is the normalized state of the array $A'$ generated by
\[
H=(-(G')^T\mid I_t).
\]
Its rows are distinct because of the identity block. Each of its last $t$ columns contains every symbol exactly $d^{t-1}$ times. For each of its first $m$ columns, Condition~(2) and Lemma~\ref{l2} give the same multiplicity. Thus $A'$ is an $\mathrm{OA}(d^t,n,d,1)$, and the displayed identity proves the asserted $LU_F$ equivalence.\qed

\noindent \textbf{Remark} For convenience, we refer to Conditions $(1)$ and $(2)$ in Theorem~\ref{t1}, for a specified strength $k$, as the $LU_F$ \emph{conditions at strength $k$}. When no strength is specified, the conditions are assumed to hold for some $k\geq1$. For a nonsystematic matrix, this terminology means that the conditions hold after the admissible conversion to systematic form described above.

\begin{example}\label{e1}
There exists an $\mathrm{OA}(16,6,2,3)$, which is $LU_F$ equivalent to an  $\mathrm{OA}(4,6,2,1)$.
\end{example}

\Proof Consider the following $4\times 6$ matrix over $\mathbb{Z}_2$:
$$
G = \begin{pmatrix}
1 & 0 & 0 & 0 & 1 & 1 \\
0 & 1 & 0 & 0 & 1 & 1 \\
0 & 0 & 1 & 0 & 0 & 1 \\
0 & 0 & 0 & 1 & 1 & 0
\end{pmatrix}=(I_4\mid G').$$

It can be checked that for any $4 \times 3$ submatrix $D$ of $G$, we have $D_3 \equiv 1 \pmod{2}$, which implies $D_3 \in \mathbb{Z}_2^*$. In addition, each row of $G'$ contains at least one nonzero element. Thus $G$ satisfies the $LU_F$ conditions and generates an $\mathrm{OA}(16, 6, 2, 3)$. Next, we prove that this $\mathrm{OA}$ is $LU_F$ equivalent to an $\mathrm{OA}(4, 6, 2, 1)$.

Let $C$ be a $2^4 \times 4$ matrix whose rows consist of all $4$-tuples over $\mathbb{Z}_2$.
Set $A=CG$. Then
$$
A = \left( \begin{array}{cccccccccccccccc}
0 & 0 & 0 & 0 & 0 & 0 & 0 & 0 & 1 & 1 & 1 & 1 & 1 & 1 & 1 & 1\\
0 & 0 & 0 & 0 & 1 & 1 & 1 & 1 & 0 & 0 & 0 & 0 & 1 & 1 & 1 & 1\\
0 & 0 & 1 & 1 & 0 & 0 & 1 & 1 & 0 & 0 & 1 & 1 & 0 & 0 & 1 & 1\\
0 & 1 & 0 & 1 & 0 & 1 & 0 & 1 & 0 & 1 & 0 & 1 & 0 & 1 & 0 & 1\\
0 & 1 & 0 & 1 & 1 & 0 & 1 & 0 & 1 & 0 & 1 & 0 & 0 & 1 & 0 & 1\\
0 & 0 & 1 & 1 & 1 & 1 & 0 & 0 & 1 & 1 & 0 & 0 & 0 & 0 & 1 & 1
\end{array} \right)^{T}
$$
  is an $\mathrm{OA}(16, 6, 2, 3)$.

The quantum state $|\psi\rangle$ corresponding to the array $A$ is
\begin{align}
\qquad|\psi\rangle = & \frac{1}{4}(| 000000 \rangle + | 000110 \rangle + | 001001 \rangle + | 001111 \rangle \nonumber\\
&+ | 010011 \rangle+ | 010101 \rangle + | 011010 \rangle + | 011100 \rangle \nonumber\\
& +| 100011 \rangle + | 100101 \rangle + | 101010 \rangle + | 101100 \rangle\nonumber\\
& + | 110000 \rangle + | 110110 \rangle + | 111001 \rangle + | 111111 \rangle).\nonumber
\end{align}

For the matrix $U_2$ defined above, $U_2 | 0 \rangle = \frac{1}{\sqrt{2}}(|0 \rangle + |1 \rangle)$,  $U_2 | 1 \rangle = \frac{1}{\sqrt{2}}(|0 \rangle - |1 \rangle)$.
We have

\begin{align}
\qquad|\phi\rangle &= U_2^{\otimes 6} |\psi\rangle\nonumber\\
&=\frac{1}{32}\sum_{i_1,i_2,\ldots,i_6=0}^{1}\Big[1+(-1)^{i_4+i_5}+(-1)^{i_3+i_6}+(-1)^{i_3+i_4+i_5+i_6}\nonumber\\
& +(-1)^{i_2+i_5+i_6}+(-1)^{i_2+i_4+i_6}+(-1)^{i_2+i_3+i_5}+(-1)^{i_2+i_3+i_4}\nonumber\\
&+(-1)^{i_1+i_5+i_6}+(-1)^{i_1+i_4+i_6}+(-1)^{i_1+i_3+i_5}+(-1)^{i_1+i_3+i_4}\nonumber\\
&+(-1)^{i_1+i_2}+(-1)^{i_1+i_2+i_4+i_5}+(-1)^{i_1+i_2+i_3+i_6}+(-1)^{\sum_{j=1}^{6}i_j}\Big]|i_1, i_2, \ldots, i_6\rangle.\nonumber
\end{align}
Denote the coefficient of $|i_1, i_2, \ldots, i_6\rangle$ by $\phi_{i_1,i_2,\ldots,i_6}$. Then
\begin{eqnarray*}
\phi_{i_1, i_2, \ldots, i_6} = \frac{1}{4}\delta_{i_1,i_2}\delta_{i_4,i_5}\delta_{i_3,i_6}[1+(-1)^{i_1+i_3+i_4}].
\end{eqnarray*}
Thus, the coefficient $\phi_{i_1, i_2, \ldots, i_6}$ is nonzero only when all three conditions $i_1 = i_2$, $i_4 = i_5$ and $i_3 = i_6$ are simultaneously satisfied. In this case, $\phi_{i_1, i_2, \ldots, i_6}=\frac{1}{4}[1+(-1)^{i_1+i_3+i_4}]$.
If the triple $(i_1,i_3,i_4)$ contains exactly one or three zeros, then $\phi_{i_1,i_2,\ldots,i_6}=\frac{1}{2}$; otherwise, this coefficient is $0$.

Consequently, $|\phi\rangle = \frac{1}{2}(| 000000 \rangle +| 001111 \rangle  + | 110110 \rangle +| 111001 \rangle )$.
The quantum state $|\phi\rangle$ corresponds to the array
$$
A' = \begin{pmatrix}
0 & 0 & 0 & 0 & 0 & 0 \\
0 & 0 & 1 & 1 & 1 & 1 \\
1 & 1 & 0 & 1 & 1 & 0 \\
1 & 1 & 1 & 0 & 0 & 1
\end{pmatrix}.$$

It follows that $A'$ is an $\mathrm{OA}(4, 6, 2, 1)$. Furthermore, since the quantum states $|\phi \rangle$ and $| \psi\rangle$ are $LU_F$ equivalent, the corresponding OAs $A'$ and $A$ are $LU_F$ equivalent as well.\qed

\begin{corollary}\label{t2}
Let $ d, m\in \mathbb{N}^+$ with $d\geq 2 $, and let $G = (I_m \mid \alpha)$ be an $m \times (m+1)$ matrix over $\mathbb{Z}_d$, where $ \alpha = (g_1, g_2, \dots, g_m)^T$, $g_i \in \mathbb{Z}_d^*$ for $i=1,2,\dots,m$.
 Then $G$ generates an $\mathrm{OA}(d^m, m+1, d, m)$ over $\mathbb{Z}_d$, and it is $LU_F$ equivalent to an $\mathrm{OA}(d, m+1, d, 1)$.
\end{corollary}
\Proof Let $D$ be an $m \times m$ submatrix of $G$. Then the possible determinant values of $D$ are $1, g_i, -g_i$ for $i = 1,2,\ldots,m$, which implies that $D_m \in \mathbb{Z}_d^*$ as $1, g_i \in \mathbb{Z}_d^*$. Thus, $G$ satisfies the $LU_F$ conditions.
By Theorem~\ref{t1}, $G$ generates an $\mathrm{OA}(d^m, m+1, d, m)$, which is $LU_F$ equivalent to an $\mathrm{OA}(d, m+1, d, 1)$.\qed

\noindent \textbf{Remark} ~For $d \geq 2$, there exist $\phi(d)^m$ matrices satisfying the conditions in Corollary~\ref{t2}, where $\phi$ is Euler's totient function, i.e., $\phi(d) = |\mathbb{Z}_d^*|$. This is because the vector $(g_1, g_2, \dots, g_m)^T$ can be chosen in exactly $\phi(d)^m$ ways with $g_i \in \mathbb{Z}_d^*$ for each $i = 1, 2, \dots, m$. Specifically, when $g_1 = \cdots = g_m = 1$, the matrix $G$ can be written as
$$
G = \begin{pmatrix}
1  & & & & 1 \\
& \ddots & & & \vdots \\
& & 1 & & 1 \\
\end{pmatrix}
.$$
Further,
$G$ generates an $\mathrm{OA}(d^m, m+1, d, m)$, which is $LU_F$ equivalent to an OA$(d, m+1, d, 1)$ given below:
\begin{equation*}
\begin{pmatrix}
0 & \cdots & 0 & 0 \\
1 & \cdots & 1 & d-1 \\
\vdots & \ddots & \vdots & \vdots \\
d-1 & \cdots & d-1 & 1 \\
\end{pmatrix}.
\end{equation*}
This result is consistent with Theorem 3.2 in \cite{Du}.

\begin{theorem}\label{t3}
Let $d\geq2$, $1\leq m<n$, and let $G=(I_m\mid G')$ satisfy the $LU_F$ conditions over $\mathbb{Z}_d$. Put $t=n-m$, and let $A$ be generated by $G$. The state $U_d^{\otimes n}|\psi_A\rangle$ is precisely the normalized equal-amplitude state of the array generated by
\[
H=(-(G')^T\mid I_t).
\]
Among arrays with distinct rows, this Fourier image determines its array uniquely up to row order. In particular, the arrays generated by $G$ and $H$ are $LU_F$ equivalent.
\end{theorem}
\Proof The proof of Theorem~\ref{t1} shows that the Fourier image has coefficient $d^{-t/2}$ exactly on the vectors $(-b(G')^T,b)$, with $b\in\mathbb{Z}_d^t$, and coefficient zero elsewhere. These are precisely the distinct rows generated by $H$. The nonzero support of a positive equal-amplitude state determines its distinct-row array up to row order. This identifies the partner for $U_d^{\otimes n}$, not all possible partners under general $LU$  or $LU_F$ equivalence.\qed

\begin{theorem}\label{l3}
Let $m, k \in \mathbb{N}^+$ with $k\leq m$, and let $G = (I_m \mid G')$ be an $m \times 2m$ matrix over $\mathbb{Z}_2$ satisfying the $LU_F$ conditions at strength $k$.
If $G'= (G')^T$, then the $\mathrm{OA}(2^m, 2m, 2, k)$ generated by $G$ is $LU_F$ equivalent to an $\mathrm{OA}(2^m, 2m, 2, k)$. Furthermore, the two arrays are combinatorially equivalent.
\end{theorem}
\Proof Let $A$ be the $\mathrm{OA}(2^m, 2m, 2, k)$ generated by $G$.
By Theorem~\ref{t1}, there exists an OA$(2^m, 2m, $ $2, 1)$, denoted by $A'$, which is $LU_F$ equivalent to $A$.

By Theorem~\ref{t3}, $A'$ has a generator matrix $H =(-(G')^T \mid I_m)$.
Since $G'=(G')^T=-(G')^T$, the matrices $G$ and $H$ are equivalent under column permutations.  Therefore, the OAs $A$ and $A'$ are equivalent. Further, $A'$ is also an $\mathrm{OA}(2^m, 2m, 2, k)$.\qed

\begin{theorem}\label{l4}
Let $d, m, k \in \mathbb{N}^+$ with $k\leq m$, and let $G = (I_m \mid G')$ be an $m \times 2m$ matrix over $\mathbb{Z}_d$, with $d\geq2$, satisfying the $LU_F$ conditions at strength $k$. If $G'$ is invertible over $\mathbb{Z}_d$ and $(G')^{-1} = -(G')^T$,
then the $\mathrm{OA}(d^m, 2m, d, k)$ generated by $G$ is $LU_F$ equivalent to an $\mathrm{OA}(d^m, 2m, d, k)$. Furthermore, the two arrays are combinatorially equivalent.
\end{theorem}
\Proof Let $A$ be generated by $G$. By Theorem~\ref{t3}, its Fourier partner is generated by $H=(-(G')^T\mid I_m)$. The hypothesis gives
\[
(G')^{-1}G=((G')^{-1}\mid I_m)=H.
\]
Since left multiplication by $(G')^{-1}$ bijectively reparametrizes the input vectors, $G$ and $H$ generate exactly the same row set. Thus the arrays agree up to row order, and their common state is fixed by $U_d^{\otimes2m}$. In particular, both have strength $k$.\qed

\begin{lemma}\label{t4}
For $i=1,2$, let $G_i=(I_{m_i}\mid G_i')$ be an $m_i\times n_i$ matrix over $\mathbb{Z}_d$ satisfying the $LU_F$ conditions at strength $k_i$, and put $t_i=n_i-m_i$. Let $A_i$ and $A_i'$ be generated by $G_i$ and $H_i=(-(G_i')^T\mid I_{t_i})$, respectively. Then
\[
G=\begin{pmatrix}G_1&0\\0&G_2\end{pmatrix}
\]
satisfies the $LU_F$ conditions at every strength $1\leq k\leq\min\{k_1,k_2\}$ after reordering its coordinate blocks to obtain systematic form. In the original coordinate order, the array generated by $G$ has Fourier partner generated by
\[
H=\begin{pmatrix}H_1&0\\0&H_2\end{pmatrix}.
\]
\end{lemma}
\Proof See Appendix~\ref{app-A}.\qed

The corresponding upper-triangular construction also preserves the stated range of strengths.
\begin{lemma}\label{l13}
Under the hypotheses on $G_1$ and $G_2$ in Lemma~$\ref{t4}$, let $B$ be any $m_1\times n_2$ matrix over $\mathbb{Z}_d$. Then
\[
G=\begin{pmatrix}G_1&B\\0&G_2\end{pmatrix}
\]
can be put in systematic form satisfying the $LU_F$ conditions at every strength $1\leq k\leq\min\{k_1,k_2\}$.
\end{lemma}
\Proof  See Appendix~\ref{app-B}.\qed

The correspondence between OAs and quantum states yields the following result.

\begin{theorem}\label{c1}
Let the quantum states $|\psi_1\rangle$, $|\psi_2\rangle$ and $|\psi\rangle$ be given by the OAs $A_1$, $A_2$ and $A$ in Lemma~$\ref{t4},$ respectively. Then we have

 $(1)$ $|\psi\rangle = |\psi_1\rangle \otimes |\psi_2\rangle;$

 $(2)$ If $|\psi_1\rangle$ and $|\phi_1\rangle$ are $LU_F$ equivalent, and $|\psi_2\rangle$ and $|\phi_2\rangle$ are $LU_F$ equivalent, then $|\psi\rangle$ and $|\phi\rangle$ are $LU_F$ equivalent, where $|\phi\rangle = |\phi_1\rangle \otimes |\phi_2\rangle$.
\end{theorem}
\Proof See Appendix~\ref{app-C}.\qed
\section{LU equivalence of IrOAs}
In this section, we study $LU$  equivalence of irredundant orthogonal arrays (IrOAs).
IrOAs of strength $k$ can be used to construct $k$-uniform states  \cite{Goyeneche2014}. In addition,
 $LU$ operations preserve $k$-uniformity. Thus, $LU$ equivalent states associated with IrOAs have the same uniformity properties. We present several families of IrOAs that are $LU_F$ equivalent to IrOAs of strength $1$ or $2$.

We first introduce the following notation.
Let $S$ be a finite set and $S^l= \{(v_1,v_2,\ldots, v_l) : v_i \in S, i = 1,2,\ldots,l\}$. The \emph{Hamming distance} HD$(u,v)$ between the vectors $u=(u_1,u_2,\ldots, u_l)$, $v=(v_1,v_2,\ldots, v_l)\in S^l$ is defined as the number of positions in which they differ. The \emph{minimum distance} of a matrix $A$, written MD$(A)$, is defined to be the minimum Hamming distance between rows at distinct row indices. HD$(A)$ denotes the set of Hamming distances between rows at distinct row indices in $A$.

\begin{lemma}{\rm(\!\!\cite{Pang})}\label{l5}
An $\mathrm{OA}(r, N, d, k)$ is irredundant if and only if its minimum distance is greater than $k$.
\end{lemma}

The next observation supplies the converse needed when transferring irredundancy through an $LU$  transformation.

\begin{lemma}\label{positive-uniform}
Let $A$ be an $r\times N$ array over $[d]$ with distinct rows, and let $|\psi_A\rangle$ be its normalized positive equal-amplitude state. For $1\leq k<N$, the state $|\psi_A\rangle$ is $k$-uniform if and only if $A$ is an $\mathrm{IrOA}(r,N,d,k)$.
\end{lemma}
\Proof
 Consider a bipartition ($T,\overline{T}$) of the $N$ coordinates, where $|T|=k$ and $\overline{T}$ is the complement of $T$  with $|\overline{T}|=N-k$. For any row $a\in A$, let $a_T$ and $a_{\overline{T}}$ denote the sets of elements of $a$ in coordinates $T$ and $\overline{T}$, respectively. The reduced density matrix $\rho_T$ on any $k$ coordinates $T$ can be expresses as
\begin{equation*}
\rho_T=\frac{1}{r}\sum_{a,b \in A} |a_T\rangle\langle b_T|\times  \langle a_{\overline{T}}|b_{\overline{T}}\rangle.
\end{equation*}
For $a_T=b_T$, distinctness of the full rows forces $a=b$, so the diagonal entries are all $d^{-k}$ exactly when the projection onto $T$ is uniform. The off-diagonal entries are nonnegative and vanish exactly when no two distinct rows agree on $\overline{T}$. Requiring both properties for every $T$ is precisely the definition of an IrOA.\qed

In particular, if an $LU$  image of an IrOA state is again a positive equal-amplitude state with distinct support, its support is an IrOA of the same strength. Positivity is essential to this converse because it rules out cancellation of off-diagonal terms.

\begin{theorem}\label{t5}
Let $d, m, n \in \mathbb{N}^+$ with $d\geq2$, $m<n$, and let $G = (I_m \mid G')$ be an $m \times n$ matrix over $\mathbb{Z}_d$ satisfying the $LU_F$ conditions. Then $G$ generates an $\mathrm{IrOA}(d^m, n, d, 1)$, which is $LU_F$ equivalent to an $\mathrm{IrOA}(d^t, n, d, 1)$, where $t = n - m$.
\end{theorem}

\Proof Let $A$ be generated by $G$, which satisfies the $LU_F$ conditions at some strength $k\geq1$. By Theorem~\ref{t1}, $A$ is an OA of strength at least $1$, and its Fourier partner $A'$ is an $\mathrm{OA}(d^t,n,d,1)$ with distinct rows.

Consider a nonzero difference $z\in\mathbb{Z}_d^m$ between two input vectors. Their corresponding rows differ by $(z,zG')$. If $z$ has at least two nonzero coordinates, their Hamming distance is at least $2$. Otherwise, let $z_i$ be its unique nonzero coordinate. Some entry $g'_{ij}$ is a unit, so $z_i g'_{ij}\neq0$. The rows then differ both in coordinate $i$ and in a coordinate of the final block. Hence $\mathrm{MD}(A)\geq2$, and Lemma~\ref{l5} makes $A$ an IrOA of strength $1$.

Its state is therefore $1$-uniform. Since $LU$  transformations preserve $1$-uniformity, Lemma~\ref{positive-uniform} shows that $A'$ is also an IrOA of strength $1$.\qed

\begin{example}
The OAs $A$ and $A'$ in Example~$\ref{e1}$ have minimum distances $2$ and $4$, respectively. By Lemma~$\ref{l5}$, $A$ is an $\mathrm{IrOA}(16,6,2,1)$ and $A'$ is an $\mathrm{IrOA}(4,6,2,1)$. They are $LU_F$ equivalent, in agreement with Theorem~$\ref{t5}.$
\end{example}
\begin{theorem}\label{t6}
Let $d,m,k\in\mathbb{N}^+$ with $d\geq 2$, $m\geq 4$, and $2\leq k\leq m$, and let $G = (I_m \mid G')$ be an $m \times (2m-1)$ matrix over $\mathbb{Z}_d$. If $G$ satisfies the following conditions:

$(1)$ $D_k \in \mathbb{Z}_d^*$ for every $m \times k$ submatrix $D$ of $G;$

 $(2)$ the entries of $G'$ belong to $\{0\}\cup\mathbb{Z}_d^*$, and each column of $G'$ contains exactly one $0$, while each row contains at most one $0$.
\\
Then $G$ generates an $\mathrm{IrOA}(d^m, 2m-1, d, 2)$, which is $LU_F$ equivalent to an $\mathrm{IrOA}(d^{m-1}, \\2m-1, d, 2)$.
\end{theorem}
\Proof See Appendix~\ref{app-D}.\qed

\noindent \textbf{Remark}
Under the hypotheses of Theorem~\ref{t6}, there exists an OA$(d^m, 2m-1, d, k)$, which is $LU_F$ equivalent to an OA$(d^{m-1}, 2m-1, d, 2)$, where $m, k \in \mathbb{N}^+$ and $2\leq k \leq m$.

\begin{theorem}
For $d, m \in \mathbb{N}^+$ with $d \geq 2$, $m \geq 4$, there exists an $\mathrm{IrOA}(d^m, 2m-1, d, 2)$ over $\mathbb{Z}_d$, which is $LU_F$ equivalent to an $\mathrm{IrOA}(d^{m-1}, 2m-1, d, 2)$.
\end{theorem}

\Proof Let $G=(I_m\mid G')$ be an $m \times (2m-1)$ matrix over $\mathbb{Z}_d$, where
$$
G' = \begin{pmatrix}
0 & 1 & \cdots & 1\\
1 & 0 & \cdots & 1\\
\vdots & \vdots & \ddots & \vdots\\
1 & 1 & \cdots & 0\\
1 & 1 & \cdots & 1
\end{pmatrix}.$$

We verify Condition~(1) of Theorem~\ref{t6} with $k=3$. Write the identity columns as $e_1,\ldots,e_m$ and the remaining columns as $v_j=\textbf{1}_m-e_j$, $1\leq j\leq m-1$. Every selection of three columns has a $3\times3$ minor equal to $1$ or $-1$:
\begin{itemize}
\item For three identity columns, use their three corresponding rows.
\item For two identity columns and one $v_j$, choose an additional row outside the two identity positions at which $v_j$ is $1$. Such a row exists because $m\geq4$.
\item For $e_a,v_b,v_c$ with $a,b,c$ distinct, use rows $a,b,c$. If $a=b$ or $a=c$, use these two distinct positions and any third position.
\item For $v_a,v_b,v_c$ with distinct indices, use rows $a,b,h$, where $h\notin\{a,b,c\}$; the resulting minor has determinant $1$.
\end{itemize}
Thus every three-column submatrix has a unit determinantal divisor for every $d\geq2$. The zero pattern also satisfies Condition (2). Theorem~\ref{t6} therefore gives an $\mathrm{OA}(d^m,2m-1,d,$ $3)$ that is an $\mathrm{IrOA}(d^m,2m-1,d,2)$, with an $LU_F$ equivalent $\mathrm{IrOA}(d^{m-1},2m-1,d,2)$.\qed

As in Theorem~\ref{t6}, for $d\geq 2$ and $m\geq 3$, we obtain $m\times 2m$ generator matrices whose IrOAs are $LU_F$ equivalent to IrOAs with the same parameters.

\begin{theorem}\label{t7}
Let $d,m,k\in\mathbb{N}^+$ with $d\geq 2$, $m\geq 3$, and $2\leq k\leq m$, and let $G = (I_m \mid G')$ be an $m \times 2m$ matrix over $\mathbb{Z}_d$. If $G$ satisfies the following conditions:

 $(1)$ for every $m \times k$ submatrix $D$ of $G$, $D_k \in \mathbb{Z}_d^*;$

 $(2)$ the entries of $G'$ belong to $\{0\}\cup\mathbb{Z}_d^*$, while each column and each row contains exactly one $0$.
\\
 Then $G$ generates an $\mathrm{IrOA}(d^m, 2m, d, 2)$ over $\mathbb{Z}_d$, which is $LU_F$ equivalent to an $\mathrm{IrOA}(d^m, 2m, d, 2)$.
\end{theorem}
\Proof See Appendix~\ref{app-E}.\qed

The following example gives two such arrays with different row sets. This does not by itself establish combinatorial inequivalence.
\begin{example}
There exists an $\mathrm{IrOA}(3^4, 8, 3, 2)$ over $\mathbb{Z}_3$, which is $LU_F$ equivalent to an $\mathrm{IrOA}(3^4, 8,$ $ 3, 2)$.
\end{example}

\Proof Let $d = 3$, $m = 4$, and let $G_1$ be the following $4\times 8$ matrix over $\mathbb{Z}_3$:
$$
G_1 = \begin{pmatrix}
1 & 0 & 0 & 0 & 0 & 1 & 2 & 2\\
0 & 1 & 0 & 0 & 2 & 0 & 2 & 1\\
0 & 0 & 1 & 0 & 2 & 2 & 0 & 2\\
0 & 0 & 0 & 1 & 1 & 2 & 2 & 0
\end{pmatrix}.$$
It is readily verified that $G_1$ meets the conditions stated in Theorem~\ref{t7}. Furthermore, any $2$ columns of $G_1$ are linearly independent, whereas some $3$ columns are linearly dependent. By Construction~\ref{c}, $G_1$ generates an $\mathrm{OA}(3^4, 8, 3, 2)$ $A$ over $\mathbb{Z}_3$, and there exists an OA $A'$ that is $LU_F$ equivalent to $A$.
In light of Theorem~\ref{t3}, $A'$ is generated by the following $4 \times 8$ matrix
$$
G_2 = \begin{pmatrix}
0 & 1 & 1 & 2 & 1 & 0 & 0 & 0\\
2 & 0 & 1 & 1 & 0 & 1 & 0 & 0\\
1 & 1 & 0 & 1 & 0 & 0 & 1 & 0\\
1 & 2 & 1 & 0 & 0 & 0 & 0 & 1
\end{pmatrix}.$$

We can check that any $2$ columns of $G_2$ are linearly independent and some $3$ columns are linearly dependent, which implies that $A'$   is an $\mathrm{OA}(3^4, 8, 3, 2)$. According to Theorem~\ref{t7}, both $A$ and $A'$ are $\mathrm{IrOAs}$ of strength $2$. Moreover, the vector $(0,0,0,1,1,2,2,0)$ is a row of $A$, but it is not contained in $A'$. So, $A$ and $A'$ are different.
\qed
\section{LU equivalence of linear OAs and codes duality}

The rows of an orthogonal array can be interpreted as codewords. Conversely, the dual distance of a linear code determines the strength of its associated array. We now relate the Fourier constructions in Section 2 to this correspondence.

The coding-theoretic definitions below apply to prime powers $q$. However, whenever we apply the cyclic Fourier matrix $U_q$ from Section 2 to codes over $\mathbb{GF}(q)$, we assume that $q$ is prime, so that $\mathbb{GF}(q)=\mathbb{Z}_q$. For nonprime prime powers, the additive group of $\mathbb{GF}(q)$ is not the cyclic group $\mathbb{Z}_q$; a field-additive Fourier transform would be needed instead. We do not identify these two transforms.

Let $\mathbb{GF}(q)^n$ be an $n$-dimensional row vector space over the finite field $\mathbb{GF}(q)$, where $q$ is a prime power.  A vector subspace $\mathcal{C}$ of $\mathbb{GF}(q)^n$ is called a $q$-\emph{ary linear code}. The vectors in $\mathcal{C}$ are called \emph{codewords}. The \emph{minimum distance} of a code $\mathcal{C}$ is the smallest Hamming distance between any two distinct codewords, denoted by $d(\mathcal{C})$. Thus, a $q$-ary linear code $\mathcal{C}$ has three basic parameters: code length $n$, dimension $m = \log_q |\mathcal{C}|$ and minimum distance $d = d(\mathcal{C})$.  We refer to such a code as an $[n,m,d]_q$ \emph{code}. When $q=2$, the subscript is usually omitted.

For an $[n, m, d]_q$ code $\mathcal{C}$,  we take a basis $\{c_1, c_2, \dots, c_m\}$ of $\mathcal{C}$ in $\mathbb{GF}(q)^n$.
Each codeword $c$ of $\mathcal{C}$ can be written uniquely as
\begin{eqnarray*}
c = x_1 c_1 + x_2 c_2 + \cdots + x_m c_m = (x_1, x_2, \dots, x_m)~ G,
\end{eqnarray*}
where $x_1, x_2, \dots, x_m \in \mathbb{GF}(q)$, and
$$G = \begin{pmatrix} c_1\\ c_2\\ \vdots\\ c_m \end{pmatrix}.$$
 We refer to $G$ as a \emph{generator matrix} of the linear code $\mathcal{C}$. Moreover,  define
\begin{eqnarray*}
\mathcal{C}^\perp = \{ v \in \mathbb{GF}(q)^n : cv^T=0\text{ for every }c\in\mathcal{C} \}.
\end{eqnarray*}
Then $\mathcal{C}^\perp$ is also a linear subspace of $\mathbb{GF}(q)^n$, $\dim \mathcal{C} + \dim \mathcal{C}^\perp = n$ and $(\mathcal{C}^\perp)^\perp=\mathcal{C}$. Thus, $\mathcal{C}^\perp$ is a linear code with parameters $[n, n-m, d^\perp]_q$, called the \emph{dual} of $\mathcal{C}$. $d^\perp$ is the minimum distance of $\mathcal{C}^\perp$, called the \emph{dual distance} of $\mathcal{C}$. A generator matrix of $\mathcal{C}^\perp$ is called the \emph{parity-check matrix} of $\mathcal{C}$.

\begin{lemma}{\rm(\!\!\cite{Feng})}\label{l7}
Let $G=(I_m\mid G')$ generate an $[n,m,d]_q$ code $\mathcal{C}$ with $0<m<n$. Then $\mathcal{C}$ has a parity-check matrix $H = (-(G')^T \mid I_{n-m})_{(n-m)\times n}$. Further,   $\mathcal{C}$ has minimum distance $d$ if and only if every $d-1$ columns of $H$ are linearly independent over $\mathbb{GF}(q)$, and some $d$ columns are linearly dependent.
\end{lemma}

For a linear code $\mathcal{C}$ of dimension $m$, the array of all $q^m$ codewords is called a \emph{linear array}. When it has positive OA strength, it is called a \emph{linear OA}. A generator matrix of the code also generates its array, and a parity-check matrix of the array means a parity-check matrix of the code. A general linear code need not yield an OA of positive strength: a zero coordinate in the code gives a constant column. In statements involving both a code and its dual, we assume $0<m<n$ so that both minimum distances are defined. When the dual distance is $1$, strength $0$ denotes only the vacuous empty-column condition.

\begin{lemma}{\rm(\!\!\cite{Hedayat1999})}\label{ca}
If $\mathcal{C}$ is an $[n, m, d]_q$ code over $\mathbb{GF}(q)$ with dual distance $d^{\bot}$, then the codewords of $\mathcal{C}$ form the rows of an {\rm OA}$(q^m, n, q, d^{\perp}-1)$ over $\mathbb{GF}(q)$. Conversely, the rows of a linear {\rm OA}$(q^m,n,q,t)$ over $\mathbb{GF}(q)$ form an $[n, m, d]_q$ linear code over $\mathbb{GF}(q)$ with dual distance $d^{\perp} \geq t+1$. If the orthogonal array has strength $t$ but not $t+1$, $d^{\perp}$ is precisely $t+1$.
\end{lemma}

Using the equivalence between irredundancy and minimum distance, Bajalan and Boyvalenkov investigated how linear codes and their duals can give rise to IrOAs \cite{Bajalan}.  They also constructed new families of IrOAs from well-known classes of codes.

\begin{lemma}{\rm(\!\!  \cite{Bajalan})}\label{IrOA}
 Let $\mathcal{C}$ be an $[n, m, d]_q$ code with $0<m<n$, $d\geq2$, and $d^{\perp}\geq2$, and let $\mathcal{C}^{\perp}$
  be its dual code. Consider the arrays at their respective maximal strengths $d^{\perp}-1$ and $d-1$. Then $n\geq2\min\{d-1,d^{\perp}-1\}$, and the following hold:

$(1)$ The {\rm OA} from $\mathcal{C}$ $($or $\mathcal{C}^{\perp})$ is an {\rm IrOA} if and only if $d\geq d^{\perp}$ $($or $d^{\perp}\geq d); $

~~~~In particular, if $d^{\perp}\neq d$, then exactly one of the {\rm OAs} from $\mathcal{C}$ and $\mathcal{C}^{\perp}$ respectively
 is an  {\rm IrOA} $($the one with greater minimum distance$)$.

$(2)$ The {\rm OAs} from $\mathcal{C}$ and $\mathcal{C}^{\perp}$ respectively are {\rm IrOAs} if and only if $d^{\perp}= d$.
\end{lemma}

At a common, possibly smaller strength $t$, Lemmas~\ref{l5} and~\ref{ca} give the separate criterion
\[
\text{both arrays are IrOAs of strength }t
\quad\Longleftrightarrow\quad
1\leq t\leq\min\{d-1,d^{\perp}-1\}.
\]
This distinction between maximal and prescribed strengths will be used below.

From Theorem~\ref{t3} and the relationship between linear OAs and linear codes, we obtain the following result.

\begin{theorem}\label{c2}
Let $m,n\in\mathbb{N}^+$ with $m<n$, and $q$ be a prime. Let $G = (I_m \mid G')$ be a matrix over $\mathbb{GF}(q)$ satisfying the $LU_F$ conditions. Then the linear {\rm OA} $A$ generated by $G$ is $LU_F$ equivalent to an {\rm OA} $A'$ generated by $H$, where $H$ is the parity-check matrix of $A$. Moreover, the linear codes corresponding to $A$ and $A'$ are duals.
\end{theorem}

\Proof By Theorem~\ref{t3}, the Fourier partner of $A$ is generated by $H=(-(G')^T\mid I_{n-m})$. Since $GH^T=0$ and $H$ has rank $n-m$, its row space is exactly the dual of the row space of $G$. Lemma~\ref{l7} identifies $H$ as a parity-check matrix.\qed

We next study $LU$  equivalence of OAs arising from
 MDS codes, Hamming codes, Reed--Muller codes, cyclic codes and negacyclic codes, respectively.

\subsection{LU equivalence of OAs from MDS codes}

By the Singleton bound, if an $[n, m, d]_q$ code exists,  then $n \geq m + d - 1$. In particular, a code that meets this bound, i.e., $n = m + d-1$, is called a \emph{maximum distance separable} (MDS) \emph{code}.

\begin{lemma}{\rm (\!\! \cite{Feng})}\label{l8}
Let $\mathcal{C}$ be an $[n,m,d]_q$ code, and let $G$ and $H$ be its generator matrix and parity-check matrix, respectively. Then the following four conditions are equivalent:

 $(1)$ $\mathcal{C}$ is an {\rm MDS}  code.

$(2)$ Any $m$ columns of $G$ are linearly independent.

$(3)$ Any $n-m$ columns of $H$ are linearly independent.

$(4)$  $\mathcal{C}^\perp$ is an   {\rm MDS}  code.
\end{lemma}

From Lemma~\ref{l8}, if $\mathcal{C}$ is an $[n, m, d]_q$ MDS code, then $\mathcal{C}^\perp$ is an $[n, n-m, m+1]_q$ MDS code.
 The MDS property yields the following result.

\begin{lemma}\label{mds-oa}
Let $m,n\in\mathbb{N}^+$ with $m<n$ and $q$ be a prime, and let $G = (I_m \mid G')$ be an $m \times n$ matrix over $\mathbb{GF}(q)$. If any $m$ columns of $G$ are linearly independent, then $G$ generates an $\mathrm{OA}(q^m, n, q, m)$, which is $LU_F$ equivalent to an $\mathrm{OA}(q^t, n, q, t)$, where $t = n - m$. The two {\rm OAs}   are both irredundant at their respective $\mathrm{OA}$ strengths $m$ and $t$ if and only if $m=t$.
\end{lemma}

\Proof Every entry $g'_{ij}$ is nonzero: replacing identity column $i$ by column $j$ of $G'$ gives an $m\times m$ minor equal to $\pm g'_{ij}$, which must be nonzero. Since $q$ is prime, these entries are units. Together with the assumed independence of every $m$ columns, this verifies the $LU_F$ conditions at strength $m$. Then $G$ generates a linear $\mathrm{OA}(q^m, n, q, m)$ $A$, which is $LU_F$ equivalent to an $\mathrm{OA}(q^t, n, q, 1)$ $A'$ by Theorem~\ref{t1}, and $A'$ has a generator matrix $H=(-(G')^{T} \mid I_t)$ by Theorem~\ref{t3}.

 By Theorem~\ref{c2}, we see that the linear codes corresponding to $A$ and $A'$ are duals. According to Lemma~\ref{l8}, these two codes are both MDS codes, and any $t$ columns of $H$ are linearly independent, where $t = n - m$. This implies that the OA $A'$ generated by $H$ is an $\mathrm{OA}(q^t, n, q, t)$. The code generated by $G$ has minimum distance $t+1$, and its dual has minimum distance $m+1$. By Lemma~\ref{l5}, the first array is irredundant at strength $m$ exactly when $t\geq m$, while the second is irredundant at strength $t$ exactly when $m\geq t$. Thus both are irredundant at these respective strengths if and only if $m=t$.\qed

\begin{theorem}\label{t8}
For a prime $q$ and $m,n\in \mathbb{N}^+$ with $1 \leq m<n \leq q$, there exists an $\mathrm{OA}(q^m, n, q, m)$, which is $LU_F$ equivalent to an $\mathrm{OA}(q^t, n, q, t)$, where $t = n - m$. The two $\mathrm{OA}$s are both irredundant at their respective $\mathrm{OA}$ strengths $m$ and $t$ if and only if $m=t$.
\end{theorem}

\Proof Let $G$ be an $m \times n$ Vandermonde matrix over $\mathbb{GF}(q)$,
$$
G = \begin{pmatrix}
1 & 1 & \cdots & 1 \\
\alpha_1 & \alpha_2 & \cdots & \alpha_n \\
\alpha_1^2 & \alpha_2^2 & \cdots & \alpha_n^2 \\
\vdots & \vdots & \ddots & \vdots \\
\alpha_1^{m-1} & \alpha_2^{m-1} & \cdots & \alpha_n^{m-1}
\end{pmatrix},
$$
where $\alpha_1,\alpha_2,\dots,\alpha_n$ are distinct elements of $\mathbb{GF}(q)$ and $1\leq m<n\leq q$.
The determinant of any $m$-column submatrix is a product of nonzero differences between evaluation points, so every such submatrix is invertible. Multiplying $G$ on the left by the inverse of its first $m$ columns gives a systematic matrix $(I_m\mid G')$. By Lemma~\ref{mds-oa}, $G$ generates an $\mathrm{OA}(q^m,n,q,m)$ whose Fourier partner is an $\mathrm{OA}(q^t,n,q,t)$, with $t=n-m$. The two arrays are both irredundant at their respective maximal strengths if and only if $m=t$.\qed

\subsection{LU equivalence of OAs from Hamming codes}

Let $m \geq 2$. The nonzero vectors in $\mathbb{GF}(q)^m$ can be partitioned into equivalence classes under the relation: two nonzero vectors $v_1$ and $v_2$ of $\mathbb{GF}(q)^m$ are linearly dependent if and only if there exists $\alpha \in \mathbb{GF}(q)^*$ such that $v_1 = \alpha v_2$. This is an equivalence relation. Each equivalence class contains exactly $q-1$ nonzero vectors, so there are $\frac{q^m - 1}{q - 1}=n$ equivalence classes. From each equivalence class, we now choose one representative $u_i$, $i=1,2, \dots, n$. Further, we arrange these $n$ vectors as columns to form an $m \times n$ matrix over $\mathbb{GF}(q)$
$$H_m = (u_1^T, u_2^T, \dots, u_n^T).$$
Note that the vectors $u_1=(1, 0, \dots, 0), u_2=(0, 1, \dots, 0), \dots, u_m=(0, 0, \dots, 1)$ belong to $m$ distinct equivalence classes. Thus,  $H_m$ can be written as $$H_m = (I_m \mid P),$$
where $P$ is an $m \times (\frac{q^m - 1}{q - 1}-m)$ matrix.
Any two columns of $H_m$ are linearly independent, because they represent distinct one-dimensional subspaces. Each row of $P$ is nonzero: for each $i$, a representative proportional to $u_i+u_j$, with $j\neq i$, occurs among its columns. In particular, when $q$ is prime, $H_m$ satisfies the $LU_F$ conditions at strength $2$.

The linear code with $H_m$ as its parity-check matrix is called a \emph{Hamming code}. A Hamming code is a $[\frac{q^m - 1}{q - 1}, \frac{q^m - 1}{q - 1} - m, 3 ]_q$ code. In particular,
a binary Hamming code has parameters $[2^m-1,2^m-m-1,3]$ and dual distance $2^{m-1}$
\cite{Feng}.

The binary Hamming-code parameters yield the following result.

\begin{theorem}\label{t9}
For any $m \geq 2$, there exists an $\mathrm{OA}(2^m, 2^m - 1, 2, 2)$, which is $LU_F$ equivalent to an $\mathrm{OA}(2^k, 2^m - 1, 2, 2^{m -1} - 1)$, where $k =2^m-m-1$. These two {\rm OAs}   are both {\rm IrOAs}  of strength $t$ if and only if $1\leq t\leq \min\{2, 2^{m -1} - 1\}$.
\end{theorem}

\Proof Let $H_m$ be an $m \times (2^m -1)$ matrix consisting of all nonzero vectors of $\mathbb{GF}(2)^m$ as its columns. Then $H_m$ can be written as $H_m = (I_m \mid P)$. Let $\mathcal{C}$ be the linear code with $H_m$ as its parity-check matrix. That is, $\mathcal{C}$ is a $[2^m -1, 2^m -m - 1,3]$ Hamming code,  and $\mathcal{C}$ has a dual code $\mathcal{C}^{\perp}$ with parameters $[2^m -1, m,2^{m-1}]$ \cite{Feng}. Since $H_m$ satisfies the $LU_F$ conditions, the OAs corresponding to linear codes $\mathcal{C}$ and $\mathcal{C}^{\perp}$ are $LU_F$ equivalent by Theorem~\ref{c2}.

By Lemma~\ref{ca}, the OA generated by $H_m$ has strength $2$. Similarly,   the OA generated by $G_m$ has strength $2^{m-1}-1$, where $G_m$ is the generator matrix of $\mathcal{C}$. Thus, the two OAs are $\mathrm{OA}(2^m, 2^m - 1, 2, 2)$ and $\mathrm{OA}(2^k, 2^m - 1, 2, 2^{m -1} - 1)$ respectively.
Further,  the two OAs are IrOAs of strength $t$ if and only if $1\leq t\leq \min\{2, 2^{m -1} - 1\}$ by Lemmas~\ref{l5} and~\ref{ca}.
\qed

\begin{example}\label{e2}
There exists an {\rm OA}$(2^3, 7, 2, 2)$, which is $LU_F$ equivalent to an $\mathrm{OA}(2^4, 7, 2, 3)$.
\end{example}

\Proof When $m =3$, the $[7, 4, 3]$ Hamming code $\mathcal{C}$ has the following parity-check matrix
$$
H_3 = \begin{pmatrix}
1 & 0 & 0 & 0 & 1 & 1 & 1 \\
0 & 1 &  0 & 1 & 0 & 1 & 1 \\
0 &  0 & 1 & 1 & 1 & 0 & 1
\end{pmatrix}.$$
We can verify that it satisfies the $LU_F$ conditions.
By Lemma~\ref{l7}, $\mathcal{C}$ has a generator matrix
$$
G_3 = \begin{pmatrix}
0 & 1 & 1 & 1 & 0 & 0 & 0 \\
1 & 0 & 1 & 0 & 1 & 0 & 0 \\
1 & 1 & 0 & 0 & 0 & 1 & 0 \\
1 & 1 & 1 & 0 & 0 & 0 & 1
\end{pmatrix}.$$
 It can be checked that any $3$ columns of $G_3$ are linearly independent,  whereas some $4$ columns are linearly dependent.  So $\mathcal{C}$ has dual distance $4$. That is, the dual code $\mathcal{C}^\perp$ is a $[7,3,4]$ code. Thus the {\rm OAs}   corresponding to linear codes $\mathcal{C}^\perp$ and $\mathcal{C}$ are {\rm OA}$(2^3, 7, 2, 2)$ and $\mathrm{OA}(2^4, 7, 2, 3)$ respectively.
Thus, the {\rm OAs}   are $LU_F$ equivalent. Moreover, they are {\rm IrOAs} of strength $2$. This result agrees with the conclusion of Theorem~\ref{t9}.\qed

\subsection{LU equivalence of OAs from Reed--Muller codes}

Reed--Muller (RM) codes have a simple recursive structure and numerous applications. Let $m \geq 1$, $0 \leq r \leq m$.  Define the base cases
$G(0,m)=(1,1,\dots,1)$, a row vector of length $2^m$, and $G(m,m)=I_{2^m}$. The binary Reed--Muller code RM$(r,m)$ is a linear code with parameters $[ 2^m, k, 2^{m-r}]$, defined recursively by the following generator matrix
$$
G(r,m) = \begin{pmatrix}
G(r,m-1) & G(r,m-1) \\
\textbf{0} & G(r-1,m-1)
\end{pmatrix},
$$
for $1\leq r<m$, where $G(r,m)$ is a $k\times 2^{m}$ matrix with $k= \sum_{i=0}^r \binom{m}{i}$ \cite{Rains,Tyagi}. In particular, when $0\leq r\leq m-1$,  RM$(r, m)^\perp=RM(m - r - 1, m)$ \cite{Feng}.

Below, we consider the OAs from binary Reed--Muller codes.
\begin{theorem}\label{t10}
For any $m \geq 2$,  there exists an $\mathrm{OA}(2^{m+1}, 2^m, 2, 3)$, which is $LU_F$ equivalent to an $\mathrm{OA}(2^k, 2^m,2,2^{m-1}-1)$, where $k =2^m-m-1$. These two {\rm OAs}   are both {\rm IrOAs} of strength $t$ if and only if $1\leq t\leq \min\{3, 2^{m -1} - 1\}$.
\end{theorem}

\Proof RM$(1, m)$ is a $[2^m, m+1, 2^{m-1}]$ code with the generator matrix
$$G = \begin{pmatrix}
    1 & 1 & \cdots & 1 \\
    0 & \multicolumn{3}{c}{} \\
    \vdots & \multicolumn{3}{c}{G_1} \\
    0 & \multicolumn{3}{c}{}
\end{pmatrix},$$
where the columns of $G_1$ consist of all nonzero vectors in $\mathbb{GF}(2)^m$, and $G_1$ is a parity-check matrix for a $[2^m - 1, 2^m - m - 1, 3]$ binary Hamming code \cite{Feng}. We will prove that $G$ satisfies the $LU_F$ conditions.

As shown in the proof of Theorem~\ref{t9}, $G_1$ has full support, and its systematic parity block has no zero row. Write $G_1=(I_m\mid P)$, where $P$ is an $m\times k$ matrix over $\mathbb{GF}(2)$.
After choosing the zero evaluation point followed by the $m$ identity columns and applying invertible row operations, $G$ can be transformed into $(I_{m+1} \mid P')$, where $P' = \begin{pmatrix} \alpha \\ P \end{pmatrix}$, $\alpha \in \mathbb{GF}(2)^k$.
The vector $\alpha$ is obtained by adding all the row vectors of $P$ to an all-ones vector of length $k$.
When $m\geq2$, $P$ has columns containing exactly two ones, so $\alpha$ has a nonzero coordinate. The remaining rows of $P'$ are nonzero because the rows of $P$ are nonzero. Also, $G$ has no zero column, since its original first row consists entirely of ones. Thus its systematic form satisfies the $LU_F$ conditions at strength $1$.

The dual is $\mathrm{RM}(m-2,m)$, with parameters $[2^m,2^m-m-1,4]$, including the repetition-code case $m=2$. Hence Lemma~\ref{ca} gives OA strengths $3$ and $2^{m-1}-1$, respectively. Theorem~\ref{c2} gives their $LU_F$ equivalence, and the common-strength criterion following Lemma~\ref{IrOA} gives exactly $1\leq t\leq\min\{3,2^{m-1}-1\}$.\qed

\noindent \textbf{Remark}
Let $G_1$ and $G$ be the matrices given in Theorem~\ref{t10}.
Denote the OAs generated by $G_1$ and $G$ by $A_1$ and $A$, respectively. Then
 $$A= \begin{pmatrix} \textbf{0}_{2^m} & A_1 \\ \textbf{1}_{2^m} & \overline{A_1} \end{pmatrix}.$$

Let $A_1=C_1G_1$ and $A=CG$, where $C_1$ is a $2^m\times m$ matrix consisting of all vectors of $\mathbb{GF}(2)^m$ as its rows; $C$ is a $2^{m+1} \times (m+1)$ matrix consisting of all vectors of $\mathbb{GF}(2)^{m+1}$ as its rows. Then we have
\begin{align}
A &= CG =  \begin{pmatrix} \textbf{0}_{2^m} & C_1 \\ \textbf{1}_{2^m} & C_1 \end{pmatrix}
\begin{pmatrix}
    1 & 1 & \cdots & 1 \\
    0 & \multicolumn{3}{c}{} \\
    \vdots & \multicolumn{3}{c}{G_1} \\
    0 & \multicolumn{3}{c}{}\end{pmatrix} \nonumber\\
    &= \begin{pmatrix} \textbf{0}_{2^m} & C_1G_1 \\ \textbf{1}_{2^m} & 1+C_1G_1 \end{pmatrix}
     = \begin{pmatrix} \textbf{0}_{2^m} & A_1 \\ \textbf{1}_{2^m} & \overline{A_1} \end{pmatrix}\nonumber,
\end{align}
where $\overline{A_1}=1+C_1G_1$ is the entrywise binary complement of $A_1$.

For this construction, the generator matrix $G$ yields an OA $A$ of strength $3$, as established in Theorem~\ref{t10}.

Let $\mathcal{C}$ and $\mathcal{C}_1$ be the linear codes whose codewords are the rows of $A$ and $A_1$, respectively. Both codes have minimum distance $2^{m-1}$, as the following weight calculation shows.

First, every nonzero codeword of $\mathcal{C}_1$ has weight $2^{m-1}$. Each row of $G_1$ corresponds to a coordinate function $x_i$. A nonzero codeword generated by $G_1$ is the evaluation of a nonzero linear function $f= a_1x_1+a_2x_2+\cdots+a_mx_m$, where $a_i\in \mathbb{GF}(2)$, $i=1,2,\dots,m$ and there exists some $i$ such that $a_i\neq0$. According to Lemma~\ref{l2}, the multiset
\begin{eqnarray*}
\{a_1 x_1 + a_2 x_2 + \cdots + a_m x_m : x_i \in \mathbb{GF}(2), i=1,2,\dots,m \} = 2^{m-1} \mathbb{GF}(2).
\end{eqnarray*}
Consequently, as the vector $(x_1,x_2,\dots,x_m)$ runs through all nonzero vectors in $\mathbb{GF}(2)^m $, the function $f$ takes the value $1$ at exactly $2^{m-1}$ positions, while it takes the value $0$ at exactly $2^{m-1}-1$ positions. This implies that the nonzero codewords in $\mathcal{C}_1$ have weight $2^{m-1}$.  Thus $\mathcal{C}_1$ has minimum distance $2^{m-1}$.

 By construction,
\begin{eqnarray*}
\mathcal{C}=\{(0,c),(1,\overline{c}):c\in\mathcal{C}_1\}.
\end{eqnarray*}
Hence the nonzero codewords in $\mathcal{C}$ have weight $2^{m-1}$ or $2^m$. Thus $\mathcal{C}$ has minimum distance $2^{m-1}$.

\begin{example}
There exists an $\mathrm{OA}(16,8,2,3)$ whose Fourier partner has the same parameters.
\end{example}

\Proof
When $m =3$, the binary {\rm RM} code, {\rm RM}$(1,3)$, has the following generator matrix
$$
 \begin{pmatrix}
1 & 1 & 1 & 1 & 1 & 1 & 1 & 1 \\
0 & 1 & 0 & 0 & 0 & 1 & 1 & 1 \\
0 & 0 & 1 &  0 & 1 & 0 & 1 & 1 \\
0 & 0 &  0 & 1 & 1 & 1 & 0 & 1
\end{pmatrix}.$$
Invertible row operations put this matrix in the form $G=(I_4\mid G')$, where
$$
G' = \begin{pmatrix}
1 & 1 & 1 & 0 \\
0 & 1 & 1 & 1 \\
1 & 0 & 1 & 1 \\
1 & 1 & 0 & 1
\end{pmatrix}.$$
It can be verified that $G$ satisfies the $LU_F$ conditions, and any $3$ columns of $G$ are linearly independent, while some $4$ columns of $G$ are linearly dependent.
By Lemma~\ref{l7}, this code has a parity-check matrix
$$
H=\begin{pmatrix}
1 & 0 & 1 & 1 & 1 & 0 & 0 & 0 \\
1 & 1 & 0 & 1 & 0 & 1 & 0 & 0 \\
1 & 1 & 1 & 0 & 0 & 0 & 1 & 0 \\
0 & 1 & 1 & 1 & 0 & 0 & 0 & 1
\end{pmatrix}.$$
One can verify that any $3$ columns of $H$ are linearly independent, whereas some $4$ columns of $H$ are linearly dependent.

Thus, the {\rm OAs} generated by $G, H$ are $LU_F$ equivalent. Moreover,   they are also {\rm IrOAs} of strength $3$. This result agrees with Theorem~\ref{t10}.\qed

\begin{theorem}\label{RM}
For $r,m\in\mathbb{N}^+$ with $r<m$, there exists an $\mathrm{OA}(2^k,2^m,2,2^{r+1}-1)$, which is $LU_F$ equivalent to an $\mathrm{OA}(2^{k'}, 2^m,2,2^{m-r}-1)$, where $k=\sum_{i=0}^{r}\binom{m}{i}$, $k'=2^m-k$. These two {\rm OAs} are both {\rm IrOAs} of strength $t$ if and only if $1\leq t\leq \min\{2^{r+1}-1, 2^{m-r}-1\}$.
\end{theorem}

\Proof See Appendix~\ref{app-F}.\qed

\subsection{LU equivalence of OAs from cyclic and negacyclic codes}
Cyclic codes are an important subclass of linear codes, characterized by closure under cyclic shifts.
A linear code $\mathcal{C}$ with parameters $[n,m,d]_q$ is called a \emph{cyclic code} if for every codeword $c =(c_0, c_1, \dots, c_{n-1})\in \mathcal{C}$, its cyclic shift $c' =(c_{n-1}, c_0, \dots, c_{n-2})$ is also in $\mathcal{C}$.

Identify each codeword $c= (c_0, c_1, \dots, c_{n-1}) \in \mathcal{C}$ with the polynomial
$$c(x) = c_0 + c_1 x + \dots + c_{n-1} x^{n-1},$$
then $\mathcal{C}$ is cyclic if and only if for every $c(x) \in \mathcal{C}$, the polynomial $x\cdot c(x)\bmod(x^n-1)$ also belongs to $\mathcal{C}$. This implies that $\mathcal{C}$ forms an ideal in the ring $\mathbb{GF}(q)[x]/(x^n-1)$. It is known that
a nonzero cyclic code has a unique monic polynomial $g(x)$ of minimum degree, called the \emph{generator polynomial} of $\mathcal{C}$.
The polynomial $h(x) = (x^n - 1)/g(x)$ is called the \emph{parity-check polynomial} of $\mathcal{C}$, where $\deg g(x) = n- m$ and $\deg h(x) = m$\cite{Feng}.

Negacyclic codes are defined analogously.
A linear code $\mathcal{C} \subseteq \mathbb{GF}(q)^n$ is called a \emph{negacyclic code} if for every codeword
$c = (c_0, c_1, \ldots, c_{n-1}) \in \mathcal{C}$, its negacyclic shift $(-c_{n-1}, c_0, c_1, \ldots, c_{n-2})$ is also in $\mathcal{C}$.

\begin{theorem}\label{t11}
Let $q$ be prime, and let $\mathcal{C}\subseteq\mathbb{GF}(q)^n$ be a cyclic or negacyclic code with $0<\dim\mathcal{C}<n$. The distinct-row arrays of $\mathcal{C}$ and its Euclidean dual are OAs of positive strength and are $LU_F$ equivalent.
\end{theorem}
\Proof Let $\lambda=1$ in the cyclic case and $\lambda=-1$ in the negacyclic case. Define the invertible shift
\[
T_\lambda(c_0,\ldots,c_{n-1})=(\lambda c_{n-1},c_0,\ldots,c_{n-2}).
\]
The code $\mathcal{C}$ is invariant under this shift. Since $\lambda^2=1$, the shift preserves the Euclidean inner product, and its inverse also preserves $\mathcal{C}$. Thus $\mathcal{C}^{\perp}$ is invariant under $T_\lambda$ as well.

Any nonzero shift-invariant code has full support. Indeed, its shifts move a nonzero coordinate through all positions, so no coordinate vanishes identically. Neither $\mathcal{C}$ nor $\mathcal{C}^{\perp}$ contains a word of weight $1$, because its shifts would span the whole ambient space, contrary to the dimension assumptions. Hence both codes have minimum distance at least $2$, and Lemma~\ref{ca} makes both associated arrays OAs of positive strength.

After a coordinate permutation and invertible row operations, choose a systematic generator $G=(I_m\mid G')$ for $\mathcal{C}$. Every column of $G$ is nonzero because $\mathcal{C}$ has full support. Every row of $G'$ is nonzero, since a zero row would give a codeword of weight $1$. Over the prime field these facts verify the $LU_F$ conditions at strength $1$. Theorem~\ref{c2} now identifies the Fourier partner with the dual-code array. Undoing the same coordinate permutation on both arrays preserves the Fourier identity in the original coordinate order.\qed

Binary Hamming codes admit cyclic realizations after a suitable coordinate ordering. For Example~\ref{e2}, any common column permutation of $G_3$ and $H_3$ preserves both code duality and the established Fourier relation; cyclicity is not needed for that example.

\section{Conclusion}

Quantum entanglement is one of the most extraordinary features of quantum physics.  $LU$  transformations can preserve the entanglement properties of quantum states.  In particular, if a state is maximally entangled, all states $LU$ equivalent to it are also maximally entangled\cite{Goyeneche2015}; likewise, if a state is $k$-uniform, every $LU$ equivalent state shares $k$-uniform property\cite{Goyeneche2014}. Understanding the interconvertibility of quantum states under $LU$  operations therefore contributes to the broader task of characterizing different classes of entangled states.
In this framework,  every column of an OA is identified with a particular qudit, and every
row corresponds to a linear term of the quantum state. Importantly, OAs are $LU$  equivalent in the sense that they lead to  $LU$  equivalent
quantum states. Quantum states constructed from OAs include diverse genuinely multipartite entangled states such
as $k$-uniform states and AME states.  Therefore, constructing $LU$  equivalent OAs is important, as it provides a concrete route to determine $LU$  equivalence among quantum states.

Goyeneche and \.{Z}yczkowski~\cite{Goyeneche2014} established a connection between OAs and multipartite quantum states, particularly between IrOAs and $k$-uniform states, and posed the problem of determining when OAs are $LU$  equivalent. Zhang et al.\cite{Zhang-Pang} provide partial solutions to the open problem. They developed methods of constructing infinite families of $LU$  equivalent OAs, both fixed-level and mixed-level via finite-field discrete Fourier transform.
In this paper, we have presented a wide range of methods for the question from constructing $LU$  equivalent OA and $LU$  equivalent IrOA classe via Fourier transforms.  In particular, for a prime $q$, an MDS generator matrix $G=(I_m\mid G')$ with $m<n$ generates an $\mathrm{OA}(q^m,n,q,m)$ that is $LU_F$ equivalent to an $\mathrm{OA}(q^{n-m},n,q,n-m)$ generated by $H=(-(G')^T\mid I_{n-m})$. These matrices are a generator matrix and a parity-check matrix of the same linear code. This connection provides constructions of $LU$  equivalent OAs from linear codes and their duals and motivates further investigation of the interplay between OA equivalence and coding theory. The results identify explicit Fourier partners rather than classify all $LU$  equivalent arrays. They also distinguish irredundancy at a prescribed common strength from irredundancy at each array's maximal strength.

Ramadas and Lakshminarayan~\cite{N-Arul} showed that phase-parametrized states associated with an IrOA$(r,N,d,k)$ yield infinitely many $LU$  equivalence classes when $r>Nd-(N-1)$. An interesting direction is to determine the number of $LU$  equivalence classes represented by the equal-amplitude states associated with OAs of prescribed parameters. Quantum orthogonal arrays (QOAs) generalize classical OAs and have been used to construct $k$-uniform states~\cite{Goyeneche2018,Zang3}. Du et al.~\cite{Du} also investigated $LU$  equivalence
of QOAs. Further constructions and classifications of QOAs under $LU$  equivalence remain promising research directions.

\section*{Declaration of competing interest}
The authors declare that they have no known competing financial interests or personal relationships that could have appeared to influence the work reported in this paper.

\section*{Data availability}
No data were used for the research described in this article.

\vspace{1cm}
\begin{center}\bfseries\Large Appendices\end{center}
\appendix

\section{Proof of Lemma~\ref{t4}}\label{app-A}

The block-diagonal generator produces exactly the Cartesian product of the row sets generated by $G_1$ and $G_2$. If a selection of $k$ columns uses $a$ columns from the first block and $b$ from the second, then $a+b=k\leq\min\{k_1,k_2\}$. The corresponding projection contains every $k$-tuple exactly
\[
d^{m_1-a}d^{m_2-b}=d^{m_1+m_2-k}
\]
times. Thus the product array has every strength in the claimed range.

Reordering the columns to place the two identity blocks first yields the systematic form
\[
\left(I_{m_1+m_2}\,\middle|\,
\begin{matrix}G_1'&0\\0&G_2'\end{matrix}\right).
\]
Every row of the right-hand block contains a unit. Thus the $LU_F$ conditions hold.

In this reordered coordinate system, Theorem~\ref{t3} gives the Fourier-partner generator
\[
\begin{pmatrix}
-(G_1')^T&0&I_{t_1}&0\\
0&-(G_2')^T&0&I_{t_2}
\end{pmatrix}.
\]
Undoing the same coordinate permutation gives $\operatorname{diag}(H_1,H_2)$ in the original coordinate order. Because $U_d^{\otimes(n_1+n_2)}$ commutes with this common permutation, the claimed Fourier relation holds in that order as well.\qed

\section{Proof of Lemma~\ref{l13}}\label{app-B}

Write $B=(B_0\mid B_1)$ according to the identity and remaining columns of $G_2$. Subtracting $B_0$ times the lower block of rows from the upper block, and then reordering columns, gives
\[
\left(I_{m_1+m_2}\,\middle|\,
\begin{matrix}G_1'&B_1-B_0G_2'\\0&G_2'\end{matrix}\right).
\]
Every row of the right-hand block contains a unit, inherited from $G_1'$ or $G_2'$.

To verify strength $k$, select $a$ columns from the first coordinate block and $b$ from the second, with $a+b=k$. For inputs $(x,y)\in\mathbb{Z}_d^{m_1}\times\mathbb{Z}_d^{m_2}$, the selected coordinates have the form $(xD_1,xB_*+yD_2)$. Each prescribed value of $xD_1$ has $d^{m_1-a}$ preimages. For each such $x$, each prescribed value of the second block has $d^{m_2-b}$ preimages in $y$. These counts also hold when $a=0$ or $b=0$. Thus every selected $k$-tuple occurs $d^{m_1+m_2-k}$ times, proving the claim by Construction~\ref{c}.\qed

\section{Proof of Theorem~\ref{c1}}\label{app-C}

Let $\mathcal{A}_i$ be the row set of $A_i$, with $|\mathcal{A}_i|=d^{m_i}$ for $i=1,2$. The row set of the block-diagonal construction is $\mathcal{A}_1\times\mathcal{A}_2$. Therefore
\begin{align*}
|\psi\rangle
&=d^{-(m_1+m_2)/2}\sum_{a\in\mathcal{A}_1}\sum_{b\in\mathcal{A}_2}|a,b\rangle\\
&=\left(d^{-m_1/2}\sum_{a\in\mathcal{A}_1}|a\rangle\right)
\otimes\left(d^{-m_2/2}\sum_{b\in\mathcal{A}_2}|b\rangle\right)
=|\psi_1\rangle\otimes|\psi_2\rangle.
\end{align*}
If $|\phi_i\rangle=V_i|\psi_i\rangle$, where $V_i$ is a tensor product of integer powers of $U_d$ on the parties of block $i$, then
\[
|\phi_1\rangle\otimes|\phi_2\rangle
=(V_1\otimes V_2)(|\psi_1\rangle\otimes|\psi_2\rangle).
\]
The operator $V_1\otimes V_2$ is again a tensor product of integer powers of $U_d$, proving part~(2). In particular, for the partners of Lemma~\ref{t4}, it is $U_d^{\otimes(n_1+n_2)}$.\qed

\section{Proof of Theorem~\ref{t6}}\label{app-D}

Let $A$ be generated by $G$. Condition~(1) gives strength $k\geq2$. The $m-1$ zero entries of $G'$ lie in distinct rows. Index these rows and columns so that the zero in column $j$ lies in row $j$  for $1\leq j\leq m-1$. Row $m$ then consists entirely of units. This is only an indexing choice; any coordinate relabeling is applied to both Fourier partners.

Two distinct rows of $A$ differ by $(z,zG')$ for a nonzero $z\in\mathbb{Z}_d^m$. Let $s$ be the number of nonzero coordinates of $z$.
\begin{itemize}
\item If $s=1$ and the nonzero coordinate is in one of the first $m-1$ rows, then $zG'$ has exactly $m-2$ nonzero coordinates. The total weight is $m-1\geq3$. If it is in row $m$, the total weight is $m\geq4$.
\item If $s=2$ and both nonzero coordinates are in rows $a,b<m$, the coordinates of $zG'$ in columns $a$ and $b$ are respectively $z_b g'_{ba}$ and $z_a g'_{ab}$. Both are nonzero because the coefficients are units. The total weight is at least $4$.
\item If $s=2$ and the nonzero coordinates are in rows $a<m$ and $m$, then column $a$ of $zG'$ equals $z_m g'_{ma}\neq0$. The total weight is at least $3$.
\item If $s\geq3$, the first $m$ coordinates already give weight at least $3$.
\end{itemize}
Thus $\mathrm{MD}(A)\geq3$, so $A$ is an $\mathrm{IrOA}(d^m,2m-1,d,2)$ by Lemma~\ref{l5}. Every row of $G'$ contains a unit. Theorem~\ref{t1} therefore gives a distinct-row Fourier partner $A'$ with $d^{m-1}$ rows. The state of $A$ is $2$-uniform, and the same is true of its $LU$  image. By Lemma~\ref{positive-uniform}, $A'$ is an $\mathrm{IrOA}(d^{m-1},2m-1,d,2)$.\qed

\section{Proof of Theorem~\ref{t7}}\label{app-E}

  Let $A$ be generated by $G$. Condition~(1) gives strength at least $2$. Consider a nonzero row difference $(z,zG')$. If $z$ has at least three nonzero coordinates, its weight is at least $3$. If $z$ has exactly one nonzero coordinate, the corresponding row of $G'$ has $m-1$ units, so the total weight is $m\geq3$. If $z$ has exactly two nonzero coordinates, say in rows $a$ and $b$, choose the column in which row $a$ has its unique zero. The entry in row $b$ of that column is a unit, so this coordinate of $zG'$ is nonzero. The total weight is again at least $3$.

Thus $\mathrm{MD}(A)\geq3$, and $A$ is an IrOA of strength $2$ by Lemma~\ref{l5}. Every row of $G'$ contains a unit, so Theorem~\ref{t1} supplies a distinct-row Fourier partner with $d^m$ rows. Its state is $2$-uniform, and Lemma~\ref{positive-uniform} makes it an $\mathrm{IrOA}(d^m,2m,d,2)$.\qed

\section{Proof of Theorem~\ref{RM}}\label{app-F}

Let $\mathcal{C}=\mathrm{RM}(r,m)$, where $1\leq r<m$. Its dimension and minimum distance are
\[
k=\sum_{i=0}^{r}\binom{m}{i},\qquad d(\mathcal{C})=2^{m-r},
\]
and its dual is $\mathrm{RM}(m-r-1,m)$, with dimension $2^m-k$ and minimum distance $2^{r+1}$~\cite{Feng,Rains,Tyagi}.

We verify the $LU_F$ conditions directly, without a nested induction on the recursive generator matrices. Put a generator of $\mathcal{C}$ in systematic form $\widetilde{G}=(I_k\mid P)$ by a coordinate permutation and invertible row operations. Since $d(\mathcal{C}^{\perp})=2^{r+1}$, every $2^{r+1}-1$ columns of $\widetilde{G}$ are independent by Lemma~\ref{l7} applied to the dual code. In particular, $2^{r+1}-1\leq k$, and Condition~(1) holds at this strength. No row of $P$ can be zero: otherwise the corresponding row of $\widetilde{G}$ would be a codeword of weight $1$, contradicting $d(\mathcal{C})=2^{m-r}\geq2$. Thus Condition~(2) also holds.

By Theorem~\ref{c2}, the codeword arrays of $\mathcal{C}$ and $\mathcal{C}^{\perp}$ are $LU_F$ equivalent. A common inverse coordinate permutation restores the original order and preserves the Fourier identity. Lemma~\ref{ca} gives their respective maximal OA strengths $2^{r+1}-1$ and $2^{m-r}-1$. Hence the arrays have parameters
\[
\mathrm{OA}(2^k,2^m,2,2^{r+1}-1),\qquad
\mathrm{OA}(2^{2^m-k},2^m,2,2^{m-r}-1).
\]
Finally, the array from $\mathcal{C}$ has strength $t$ and is irredundant precisely when $t\leq d(\mathcal{C}^{\perp})-1$ and $t\leq d(\mathcal{C})-1$. The same two inequalities apply to the dual array. Thus both are IrOAs of strength $t$ exactly for
\[
1\leq t\leq\min\{2^{r+1}-1,2^{m-r}-1\}.
\]
This includes the endpoint $r=m-1$, where the dual is the repetition code and the common IrOA strength is $1$.\qed


\begin{thebibliography}{99}
\bibitem{Bajalan}
M. Bajalan and P. Boyvalenkov. On irredundant orthogonal arrays. Discrete Appl. Math., 2026, 387: 199--208.


\bibitem{Chen}
 G. Z. Chen and X. T. Zhang. Constructions of irredundant orthogonal arrays.   Adv. Math. Commun., 2023, 17(6): 1314--1337.


 \bibitem{Du}
J. Du, C. J. Yin, S. Q. Pang and T. Y. Wang.  Local equivalence of quantum orthogonal arrays and orthogonal arrays. Quantum Inf. Process., 2020, 19: 303.

\bibitem{Dur}
W. D\"{u}r, G. Vidal and J. I. Cirac. Three qubits can be entangled in two inequivalent ways. Phys. Rev. A, 2000, 62: 062314.



\bibitem{Feng}
 K. Q. Feng. Algebraic Theory of Error-Correction Codes.  Beijing: Tsinghua University Press, 2005.



\bibitem{Goyeneche2015}
D. Goyeneche, D. Alsina, J. I. Latorre, A. Riera and K. \.{Z}yczkowski. Absolutely maximally entangled states, combinatorial designs, and multiunitary matrices. Phys. Rev. A, 2015, 92: 032316.


\bibitem{Goyeneche2018}
D. Goyeneche, Z. Raissi, S. Di Martino and K. \.{Z}yczkowski. Entanglement and quantum combinatorial designs. Phys. Rev. A, 2018, 97: 062326.

\bibitem{Goyeneche2014}
D. Goyeneche and K. \.{Z}yczkowski. Genuinely multipartite entangled states and orthogonal arrays. Phys. Rev. A, 2014, 90: 022316.

\bibitem{Hedayat1999}
 A. S. Hedayat, N. J. A. Sloane and J. Stufken. Orthogonal Arrays: Theory and Applications. Springer-Verlag, New York, 1999.


\bibitem{Kraus}
B. Kraus. Local unitary equivalence of multipartite pure states. Phys. Rev. Lett.,  2010, 104: 020504.


\bibitem{Li}
M. S. Li and Y. L. Wang. $k$-uniform quantum states arising from orthogonal arrays.  Phys. Rev. A, 2019, 99: 042332.

\bibitem{Liyu}
Y. K. Li, J. M. Yang and B. H. Zhang. Linear equations on residual class rings $\mathbb{Z}_m$. J. Shaanxi Norm. Univ. Nat. Sci. Ed.,  2004, (S1): 24--28.

\bibitem{Liu}
 B. Liu, J. L. Li, X. K. Li and C. F. Qiao. Local unitary classification of arbitrary dimensional
multipartite pure states. Phys. Rev. Lett.,  2012, 108: 050501.

\bibitem{Owen}
A. B. Owen. Orthogonal arrays for computer experiments, integration and visualization. Statistica Sinica, 1992, 2: 439--452.

\bibitem{Pang1}
 S. Q. Pang, X. Zhang, X. Lin and Q. J. Zhang. Two and three-uniform states from irredundant orthogonal arrays.  npj Quantum Inf., 2019, 5: 52.

 \bibitem{Pang}
S. Q. Pang,  R. N. Zhang and X. Zhang. Quantum frequency arrangements, quantum mixed orthogonal arrays and entangled states. IEICE Trans. Fundam., 2020, E103(12): 1674--1678.

\bibitem{Rains}
E. M. Rains and N. J. A. Sloane. Self-dual codes. Handbook of Coding Theory, Vol. 1. Amsterdam: North-Holland, 1998: 177--294.

\bibitem{Raissi}
Z. Raissi, C. Gogolin, A. Riera and A. Ac\'{i}n. Constructing optimal quantum error correcting codes from absolute maximally entangled states. J. Phys. A: Math. Theor., 2018, 51: 075301.


\bibitem{N-Arul}
N. Ramadas and A. Lakshminarayan.
Local unitary equivalence of absolutely maximally entangled states constructed from orthogonal arrays. J. Phys. A: Math. Theor., 2025, 58: 125301.


\bibitem{Rao}
C. R. Rao. Factorial experiments derivable from combinatorial arrangements of arrays. J. Roy. Statist. Soc. Suppl., 1947, 9: 128--139.


\bibitem{Tyagi}
 V. Tyagi and S. Rani. Recursive matrix method for GRM and DGRM codes. Int. Electron. J. Pure Appl. Math., 2012, 4(4): 263--270.

\bibitem{Verstraete}
F. Verstraete, J. Dehaene, B. De Moor and H. Verschelde. Four qubits can be entangled in nine different ways. Phys. Rev. A, 2002, 65: 052112.

\bibitem{xuming}
M. Xu and Z. Tian. A flexible image cipher based on orthogonal arrays. Inform. Sciences, 2021, 551: 39--53.

 \bibitem{Zang1}
 Y. Zang,  G. Chen, K. Chen and Z. Tian. Further results on 2-uniform states arising from irredundant orthogonal arrays.  Adv. Math. Commun., 2022, 16: 231--247.

\bibitem{Zang2}
 Y. Zang, P. Facchi and Z. Tian. Quantum combinatorial designs and $k$-uniform states.   J. Phys. A: Math. Theor., 2021, 54(50): 505204.

 \bibitem{Zang3}
 Y. Zang,  Z.  Tian, S. M. Fei and H. J. Zuo. Quantum $k$-uniform states from quantum
orthogonal arrays. Int. J. Theor. Phys., 2023, 62: 73.

\bibitem{Zhang}
T.  Zhang, M. J. Zhao, M. Li, S. M. Fei and X. Li-Jost. Criterion of local unitary equivalence for multipartite states. Phys. Rev. A, 2013, 88: 3367-3376.

\bibitem{Zhang-Pang}
X. Zhang, S. Q. Pang, M. Q. Chen and S. M. Fei. Construction of local unitary equivalent quantum states via local
unitary equivalent orthogonal arrays.  New J. Phys.   2026, 28: 084513.
\end{thebibliography}
\end{document}